\documentclass[aps,pra,twocolumn,nofootinbib,superscriptaddress]{revtex4}

\usepackage{mathpazo}

\usepackage{amsmath}
\usepackage{amsfonts,amssymb,amsthm, bbm,braket}
\usepackage{graphicx}   
\usepackage{subfigure}  
\usepackage{amsbsy} 
\usepackage[bold]{hhtensor} 
\usepackage{natbib}
\usepackage{algpseudocode}

\usepackage{multirow}

\usepackage{comment}

\newtheorem{definition}{Definition}
\newtheorem{corollary}{Corollary}
\newtheorem{theorem}{Theorem}
\newtheorem{lemma}{Lemma}
\newtheorem{proposition}{Proposition}

\usepackage[breaklinks=true]{hyperref}

\hypersetup{
  colorlinks   = true, 
  urlcolor     = blue, 
  linkcolor    = blue, 
  citecolor   = red 
}

\begin{document}


\title{Qubit stabilizer states are complex projective 3-designs}

\author{Richard Kueng}
\affiliation{Institute for Theoretical Physics, University of Cologne}
\affiliation{Institute for Physics $\&$ FDM, University of Freiburg}

\author{David Gross}
\affiliation{Institute for Theoretical Physics, University of Cologne}
\affiliation{Institute for Physics $\&$ FDM, University of Freiburg}

\date{\today}


\begin{abstract}
	A \emph{complex projective $t$-design} is a configuration of vectors
	which is ``evenly distributed'' on a sphere in
	the sense that
	sampling uniformly from it reproduces the moments of Haar
	measure up to order $2t$.	We show that the set of all $n$-qubit
	stabilizer states forms a complex projective $3$-design in dimension
	$2^n$. Stabilizer states had previously only been known to
	constitute $2$-designs. 
	The main technical ingredient is a general recursion formula for the
	so-called
	\emph{frame potential} of stabilizer states. To establish it, we need to
	compute the number of stabilizer states with pre-described inner
	product with respect to a reference state. This, in turn, reduces to
	a counting problem in discrete symplectic vector spaces for which we
	find a simple formula. We sketch applications in quantum information
	and signal analysis.
\end{abstract}


\maketitle

\section{Introduction and main results}

\subsection{Introduction}

In its simplest incarnation, a $D$-dimensional \emph{complex projctive $t$-design} is a set of unit-length vectors in $\mathbbm{C}^D$ that is
evenly distributed on the sphere in the sense that sampling uniformly
from this set reproduces the moments of Haar measure up to order
$2t$ \cite{delsarte_spherical_1977,renes_symmetric_2004,ambainis_wise_2007,scott_tight_2006,
matthews_distinguishability_2009}
(see \autoref{def:designs} below for a precise definition).
In a variety of contexts such a design structure is important:

In \emph{numerical integration},
designs are known as \emph{cubatures}. It follows from the definition
that the average of a homogeneous polynomial $p$ of order $2t$ over the complex
unit sphere equals $p$'s average over the design. If the design has
small order, this realization can be made the basis for fast numerical
procedures that compute integrals of smooth functions over
high-dimensional spheres.

In \emph{quantum information theory}, designs are a widely-employed
tool for \emph{derandomizing  probabilistic constructions}. 
Recall that the probabilistic method
\cite{alon_probabilistic_2004} 
is a powerful proof technique originally designed to tackle problems
in combinatorics.  At its core is the observation that the existence
of certain extremal combinatorial structures often can be be proved by
showing that a suitably chosen random construction would produce an
example with high probability. 
In quantum information, randomized construction often rely on randomly
chosen Hilbert space vectors \cite{hayden_randomizing_2004}.
While this method has brought about spectacular successes (such as the
the celebrated proof of strict sub-additivity of entanglement of
formation \cite{hastings_superadditivity_2009}),
it suffers e.g.\ from the problem that generic Haar-random states of
large quantum systems are \emph{unphysical}: they cannot be prepared
from separable inputs using a polynomial number of operations
\cite{nielsen_quantum_2010}.
Designs, in contrast, \emph{can} be chosen to consist solely of
highly-structured and efficiently preparable vectors, while retaining
``generic'' properties in a precise sense.
Thus considerable efforts have been expended at designing complex
projective designs (and their unitary cousins)
\cite{ambainis_wise_2007,dankert_exact_2009,gross_evenly_2007,low_large_2009,brandao_local_2012}.

Lastly, randomized constructions in Hilbert spaces have completely
classical applications, e.g.\ in \emph{signal analysis}. Take for
instance the highly active field of compressed sensing and related
topics \cite{fora13}: 
There, one is interested in reconstructing objects that possess some
non-trivial structure (e.g.\ sparsity, or low rank) from a small
number of linear measurements.
Strong recovery guarantees can be proven for randomly constructed
measurement vectors. Once more, this raises the problem of finding
sets of structured and well-understood measurements that sufficiently
resemble the properties of generic random vectors. The use of designs
for this purpose has been proposed in
\cite{gross_partial_2015,ehler_phase_2015,kueng_spherical_2015}.

Despite this wealth of applications and non-constructive existence proofs \cite{bondarenko_spherical_2010}, explicit constructions for
complex designs remain rare. There are varios infinite families of
complex projective 2-designs (e.g.\ maximal sets of mutually unbiased
bases \cite{klappenecker_mutually_2005, bengtsson_geometry_2006}, stabilizer
states, or symmetric informationally complete POVMs
\cite{renes_symmetric_2004}); sporadic solutions for higher orders
\cite{conway_sphere_2013, bachoc_modular_2001, gross_evenly_2007}; and approximate
constructions involving random circuits \cite{brandao_local_2012}. To the best of our
knowledge, an infinite set of explicit complex projective 3-designs
has not been identified before.

Here, we show that the set of all \emph{stabilizer states} in
dimension $2^n$ forms a complex projective 3-design for all
$n\in\mathbbm{N}$.

Recall that the stabilizer formalism is a ubiqutous tool in quantum
information theory \cite{gottesman_stabilizer_1997,nielsen_quantum_2010}. Stabilizer states (and,
slightly more general, stabilizer codes) are joint eigenvectors of
generalized Pauli matrices. Constituting the main realization of
quantum error correcting codes \cite{gottesman_stabilizer_1997}, they can be
efficiently prepared \cite{hostens_stabilizer_2005} and described in terms of
polynomially many parameters \cite{nielsen_quantum_2010}. Yet they exhibit
non-trivial properties like multi-partite entanglement
\cite{hein_entanglement_2006}. Stabilizer states were instrumental
in the development of measurement-based quantum computation
\cite{raussendorf_one_2001, gross_novel_2007}. In several precise ways,
they can be seen as the discrete analogue of Gaussian states
\cite{gross_hudsons_2006}. Beyond quantum information, stabilizer
states have proved to be versatile enough to provide powerful models
for one of the most influential recent development in theoretical
condensed mater physics: the study of topological order
\cite{kitaev_fault_2003,zeng_quantum_2015}.

Our main result thus identifies yet another aspect according to which
stabilizer states capture properties of generic state vectors.

\subsection{Designs and frame potential}

In order to state our results more precisely, we need to give a formal
definition of complex projective designs and introduce the related
notion of \emph{frame potential}. Following
\cite{scott_tight_2006,Levenshtein_designs_1998,koenig_cubature_1999},
we define

\begin{definition} \label{def:designs}
	Fix a dimension $D$ and let $\mu$ be a probability measure on the
	unit sphere in $\mathbbm{C}^D$. The measure $\mu$ is a \emph{complex
	projective $t$-design} if, for any order-$t$ polynomial $p$, we have
	\begin{equation}\label{eqn:designdef}
		\mathbb{E}_{x, y \sim \mu} 
		\left[
			p\left( | \langle x, y \rangle |^2 \right) 
		\right]
		=
		\int_{x,y} 
			p\left( | \langle x, y \rangle |^2 \right) 
			\mathrm{d} x
			\mathrm{d} y,
	\end{equation}
	where the right-hand-side integration is with respect to the uniform
	(Haar) measure on the sphere.
\end{definition}

In other words, sampling according to $\mu$ should give the same
expectation values as sampling according to the uniform measure for
any random variable that is a polynomial in $|\langle x, y
\rangle|^2$ of order at most $t$. From now on, we will only be concerned with the case where $\mu$
is the uniform measure on a finite set of unit vectors.

It is not hard to see that  $\mu$ fulfills (\ref{eqn:designdef}) for
all polynomials of order $t$ or less, if equality holds for the
specific case of $p(z)=z^{t}$. The resulting value is the $t$-th
order \emph{frame potential} 
\cite{benedetto_tight_2003}
\begin{equation}\label{eqn:potential}
	\mathcal{F}_t(\mu) := 
	\mathbb{E}_{x,y\sim \mu} 
		\left[
			 | \langle x, y \rangle |^{2t}
		\right].
\end{equation}
It is known that the Haar integral on the r.h.s.\ of
(\ref{eqn:designdef}) minimizes the frame potential over the set of
all measures $\mu$ and that, in fact, its value is given by
\begin{equation}
	\mathcal{F}_t (\mu) \geq \
	\mathcal{W}_t \left( D \right) :=
	\binom{D+t-1}{t}^{-1}.
	\label{eq:sidelnikov_inequality}
\end{equation}
This relation is known as \emph{Welch bound} \cite{welch_lower_1974} or
\emph{Sidelnikov inequality} \cite{sidelnikov_upper_1975}. In summary, we have:

\begin{theorem}[\cite{benedetto_tight_2003,scott_tight_2006,Levenshtein_designs_1998,koenig_cubature_1999}]\label{thm:designs}
	Fix a dimension $D$ and let $\mu$ be a probability measure on the
	unit sphere in $\mathbbm{C}^D$. The measure $\mu$ is a complex
	projective $t$-design if and only if its frame potential meets the
	Welch bound
	\begin{equation*}
		\mathcal{F}_t(\mu)= 
		\mathcal{W}_t(D).
	\end{equation*}
\end{theorem}

\subsection{Main results}

At the heart of this work is an explicit characterization of the frame
potential assumed by the uniform distribution over stabilizer states
in prime power dimensions $D=d^n$. We denote the set of stabilizer
states on $\big(\mathbb{C}^d\big)^{\otimes n}\simeq \mathbb{C}^D$ by
$\operatorname{Stabs}(d,n)$. The unitary symmetry group of the set of
stabilizer states is the \emph{Clifford group} (for a precise
definition, see \autoref{sub:stabilizer_states}). All results are
then implied by the following recursion formula over the dimension's
exponent $n = \log_d (D)$.

\begin{theorem}[Main Theorem]\label{thm:technical}
	Let $d$ be a prime number and let  $t \in \mathbb{N}_+$. Then for
	all dimensions $D = d^n$, the frame potential  $\mathcal{F}_t
	(\operatorname{Stabs}(d,n))$ of stabilizer states in $\mathbb{C}^D$
	is determined by the following recursion formula over $n$:
	\begin{align}
		\mathcal{F}_t (\operatorname{Stabs}(d,1)) &= \frac{d^{2-t} + 1 }{(d+1)d}, \label{eq:base_case} \\
		\frac{\mathcal{F}_t \left(\operatorname{Stabs}(d,n+1) \right)}{ \mathcal{F}_t \left( \operatorname{Stabs}(d,n) \right)}
		 &= \frac{d^{n-(t-2)} + 1}{d \left( d^{n+1} +1 \right)}.
		 \label{eq:recursion_step}
	\end{align}
\end{theorem}

Comparing this explicit characterization of the frame potential to the
Sidelnikov inequality \eqref{eq:sidelnikov_inequality} allows us to
draw the following conclusions:

\begin{corollary}\label{cor:stabs3d}
	Let $d^n$ be a prime-power dimension.
	Then the following statements are true
	\begin{enumerate}
		\item $\operatorname{Stabs}(d,n)$ forms a complex projective 2-design.
		\item $\operatorname{Stabs}(d,n)$ constitutes a complex projective
		3-design if and only if $d = 2$.
		\item 
		The set $\operatorname{Stabs}(d,n)$ does not constitute a complex
		projective 4-design.
		\item	 
		The Clifford group does not act irreducibly on
		$\operatorname{Sym}^4(\mathbb{C}^D) \subset \left( \mathbb{C}^D \right)^{\otimes 4}$. In particular, it is not a
		unitary $4$-design.
	\end{enumerate}
\end{corollary}

As indicated before, the first fact was already widely known
\cite{klappenecker_mutually_2005,bengtsson_geometry_2006,gross_evenly_2007}. The other
results, however, are new to the best of our knowledge. We reemphasize
that these assertions follow immediately form the Main~Theorem, which
may be of independent interest.

\subsection{Applications and Outlook}

Here, we sketch relations of the result to problems from signal
analysis and quantum physics. Elaborating on these connections will be
the focus of future work.

In \emph{low-rank recovery} \cite{candes_foundations_2009, gross_quantum_2010, gross_recovering_2011,
fora13}, a low-rank matrix $X$ is to be reconstructed from few linear
measurements of the form $y_i = \operatorname{tr} \left( X A_i \right)$.  In the
\emph{phase retrieval problem}
\cite{Wa63,CanStroVor13,gross_partial_2015} one aims to recover a
complex vector $x \in \mathbb{C}^D$ from the absolute value of a small
number of measurements $y_i = |\langle x, a_i \rangle|$ that are ignorant towards phase information. This task can be reduced to a particular instance of rank-one matrix recovery by
rewriting  the measurements as \cite{balan_painless_2009,candes_phase_2015}
\begin{equation*}
	y_i^2 = \operatorname{tr}\big( | x\rangle \langle x|\,|a_i\rangle
	\langle a_i|\big),
\end{equation*}
i.e.\ by setting $X=|x\rangle\langle x|$ and $A_i = |a_i\rangle\langle
a_i|$.  For both problems, strong recovery guarantees for randomly
constructed measurements are known.
Oftentimes these rely on generic (e.g.\ Gaussian) measurement ensembles
and employing complex projective designs to partially derandomize 
these result has been proposed in both contexts
\cite{gross_partial_2015,kueng_low_2015, ehler_phase_2015}.

Regarding both low rank matrix recovery and phase retrieval, it is known that
sampling measurement vectors independently from a 2-design does not do
the job \cite{gross_partial_2015}, while 4-designs
already have an essentially optimal performance \cite{kueng_low_2015, kabanava_stable_2015}.
However, the remaining intermediate case for $t=3$ is not yet fully
understood.
Numerical studies conducted by Drave and Rauhut
\cite{drave_bachelor_2015} indicate that random stabilizer-state
measurements perform surprisingly well at that task.
The combinatorial properties of prime power stabilizer states -- e.g.\
\autoref{thm:technical} -- may help to clarify this situation. We
believe this to be a potentially very insightful open problem.

Finally, we want to point out that one nice structural property of
stabilizer states is that they come in bases, i.e.\ the set of all
stabilizer states is a union of different orthonormal bases (see e.g.\ \autoref{th:stabilizer_states} below).
This allows for a considerably more structured random measurement
protocol: 
Select one such basis at random and iteratively measure the trace
inner product of an unknown low rank matrix with all projectors onto
the individual basis vectors. After having acquired $D$ data points
that way, choose a new stabilizer basis at random and repeat.  We
refer to \cite{kueng_low_2015b} for a detailed description of such a
protocol. It should be clear that it has immediate applications to
\emph{quantum state tomography}. 
In the above paper, non-trivial
recovery statements have been announced for $t$-designs that admit
such a basis structure and have strength $t \geq 3$. Again, stabilizer
states obey these criteria and have been used for the numerical
experiments conducted there. However the announced recovery
statement suffers from a non-optimal sampling rate for 3-designs and
the rich combinatorial structure of stabilizer bases 
might help to amend that situation.

\subsection{Relation to previous work and history}

After completion of this work (first announced at
the QIP 2013 conference \cite{kueng_qip_2013}), we became aware of the
fact that a close analogue of our main result follows from a statement
proved in the field of algebraic combinatorics
\cite{sidelnikov_spherical_1999} in 1999.
The object of study there is a \emph{real} version of stabilizer
states in $\mathbb{R}^{2^n}$, as well as their symmetries, which are  given
by a real version of the Clifford group. 
The key result is that under the action of the real 
Clifford group, the space
$\operatorname{Sym^3}(\mathbb{R}^{2^n})$ decomposes 
into irreps in exactly the same way as it does under the action of the
full orthogonal group $O(2^n)$
\cite{sidelnikov_spherical_1999,nebe_self_2006}. 
This implies 
\cite{goethals_spherical_1979,bannai_some_1979}
that \emph{any} orbit of the real
Clifford group gives rise to a set that reproduces moments of Haar
measure up to order $6$ (the established -- if confusing --
terminology is to refer to such sets as
\emph{spherical $6$-designs} \cite{delsarte_spherical_1977}, 
while the complex-valued analogue would be called a \emph{complex
projective $3$-design} \cite{renes_symmetric_2004}).

The findings of \cite{sidelnikov_spherical_1999} are formulated in the
language of algebraic invariant theory.  While the present authors
were trying to relate them to the results we had established in the
context of quantum information, we became aware of yet another
development. Huangjun Zhu \cite{zhu2015} independently derived a very
simple and elegant proof showing that the complex Clifford group in
dimensions $d=2^n$ actually forms a \emph{unitary 3-design}
\cite{dankert_exact_2009, gross_evenly_2007}.  This means that the the
irreducible representation spaces of the action of the Clifford group
on $\big(\mathbb{C}^{2^n}\big)^{\otimes 3}$ coincide with those of the
full unitary group $U(d)$.  In particular, the Clifford group acts
irreducibly on $\operatorname{Sym}^3(\mathbb{C}^d)$ which, in turn,
implies that that any orbit of the group constitutes a complex
projective 3-design.  The work of Zhu thus fully implies our main
result. What is more, the proof is simpler. 

The appeal of the question treated here was underscored even more,
when we learned a few days prior to submission of this paper to the
arxiv e-print server, that yet another researcher -- Zak Webb -- had
independently obtained results related to the ones of Zhu
\cite{webb2015}.

In comparision to these works,
our proof methods are completely different: We rely on counting
structures in discrete symplectic vector spaces in order to compute
the angle set between stabilizer
states, whereas \cite{sidelnikov_spherical_1999} is based on algebraic
invariant theory and \cite{zhu2015} on character theory. As a
corollary, we derive an expression for the number of stabilizer states
with prescribed inner product to a reference state.  This finding
might be of independent interest.  Also, we show that the set of
stabilizer states fails to be a 4-design in dimensions $2^n$ and that
stabilizer states in dimensions other than powers of two do not even
constitute a 3-design. The simultaneously submitted papers seem to
have left this possibility open.

\section{Proof of the main statement}

\subsection{Outline}

We already mentioned in the introduction that there is a geometric
approach to stabilizer states building on the theory of discrete
symplectic vector spaces\footnote{
This
is connected to the fact that stabilizer states are the natural discrete
analogue of \emph{Gaussian states} of bosonic systems, where the
symplectic structure is well-appreciated. For a concise introduction of this
point of view, see \cite{gross_hudsons_2006}.}.
This 
\emph{phase space formalism}
will be introduced in \autoref{sub:phase_space}.
We formally define 
stabilizer states 
and explain how to compute inner products 
in this language
in \autoref{sub:stabilizer_states}.
We then move on to briefly introducing Grassmannians and some core
concepts of discrete symplectic geometry. These tools will be used to
establish \autoref{thm:technical} in \autoref{sec:main_proof}.

\subsection{Phase Space Formalism} \label{sub:phase_space}

We start by considering a $d$-dimensional Hilbert space 
$\mathcal{H}$,
equipped with a basis $\{ | q \rangle \,|\, q\in Q\}$, where the
\emph{configuration space} $Q$ is given by 
$Q:=
	\{0, \dots, d-1\}\subset \mathbb{Z}$
with arithmetics modulo $d$.
Following 
\cite{debeaudrap_linearized_2011, appleby_symmetric_2005}, 
we define two phase factors 
$\tau := \mathrm{e}^{\pi i (d^2+1)/d}=(-1)^d \mathrm{e}^{\pi i /d}$ 
and
$\omega := \tau^2 = \mathrm{e}^{2\pi i /d}$.
For 
$q,p \in Q$, we introduce the \emph{shift} and \emph{boost} operators
defined by the relations
\begin{equation}\label{eqn:shiftboost}
\textrm{shift:} \; \hat{x}(q) |x \rangle = |x+q \rangle , \quad
\textrm{boost:} \;\hat{z}(p)|x\rangle = \omega ^{px} |x\rangle
\end{equation}
for all $x\in Q$.

For $p,q \in Q$, the corresponding \emph{Weyl operator} (or \emph{generalized Pauli operator}) is defined as
\begin{equation}\label{eqn:weyl}
w(p,q) = \tau^{-pq} \hat{z}(p) \hat{x} (q).
\end{equation}
Again following 
\cite{debeaudrap_linearized_2011, appleby_symmetric_2005}, we adopt
the convention that any artihmetic expression
\emph{in the exponent of $\tau$} is \emph{not} understood to be modulo
$d$, but rather as taking place in the integers. This makes a
difference for even dimensions (see below). One could argue that it would be
slightly cleaner to syntactically distinguish the modular operations
appearing in (\ref{eqn:shiftboost}) from the non-modular arithmetic in
(\ref{eqn:weyl}). However, the implicit convention does declutter
notation and we feel it is ultimately benefitial.

This definition is consistent with established conventions. For
example, one
recovers the usual Pauli matrices for
the qubit case $d=2$.
We use the notation $V := Q \times Q$
and consequently write $w(v) := w(v_p,v_q)$ for elements $v=(v_p,v_q) \in V$.
Furthermore we define the \emph{standard symplectic form} 
\begin{equation}
	\left[ u,v \right] := 
	u_p v_q - u_q v_p
	=
	u^T J v 
	\label{eq:symplectic}
\end{equation}	
where
\begin{equation*}
J = \left(
\begin{array}{cc} 0 & 1 \\
			-1 & 0 \end{array}
\right)
\end{equation*}
and $u=(u_p,u_q), v=(v_p, v_q) \in V$.
If $d$ is prime, 
the space $V$ together with the non-degenerate symplectic product
(\ref{eq:symplectic}) forms a symplectic vector space which is called
\emph{phase space} due to its resemblance to the phase space appearing in
classical mechanics.

The Weyl operators 
obey the \emph{composition} and \emph{commutation relations}
\begin{eqnarray}
w(u)w(v) &=& \tau^{[u,v]} w(u+v),	\label{eq:composition_relations}\\
w(u) w(v)  &=& \omega ^ {\left[u,v \right]} w(v)w(u)\quad \forall u,v \in V.	\label{eq:commutation_relations}
\end{eqnarray}
which can be verified by direct computation. 

It is worthwhile to point out that for odd $d$, the ring
$\mathbb{Z}_d$ contains a multiplicative inverse of $2$, namely
$2^{-1} = \frac{1}{2}(d+1)\in\mathbb{Z}_d$.  
This in particular assures that $\tau$ is a $d$th root of unity and
hence
the phase factors in (\ref{eqn:weyl},
\ref{eq:composition_relations}) depend only on, respectively, $pq$ and $[u,v]$
\emph{modulo $d$}. In even dimensions, however, $\tau$
has order $2d$. This somewhat complicates the theory of stabilizer
states in the even-$d$ case -- c.f.~\autoref{sub:stabilizer_states}.

The preceeding definitions have been made with a single
$d$-dimensional system in mind. We now extend our formalism to $n$
such systems.  The corresponding configuration space is $Q =
\mathbb{Z}_d^n$ with elements $q = \left(q_1 , \ldots, q_n \right)$
and $q_i \in \mathbb{Z}_d$. 
The associated phase space will be denoted by $V := Q \times Q \simeq
\mathbb{Z}_d^{2n}$ ($\dim V = 2n$).
It carries a 
symplectic form given by the natural
multi-dimensional analogue of (\ref{eq:symplectic}):
\begin{equation*}
	\left[ u,v \right] := u^T J v, 
	\qquad
	J = \left(
	\begin{array}{cc} 0_{n \times n} & \mathbb{I}_{n \times n} \\
				-\mathbb{I}_{n \times n} & 0_{n \times n} \end{array}
	\right).
\end{equation*}
With elements $(p,q)\in V$, we associate Weyl operators 
\begin{align*}
w(p,q) =& w(p_1,\ldots ,p_n,q_1,\ldots q_n) \\
=& w(p_1,q_1) \otimes \ldots
\otimes w_(p_n,q_n)
\end{align*}
acting on the tensor product space $\big(\mathbb{C}^d\big)^{\otimes n}$.
With
these definitions, the composition and commutation
relations (\ref{eq:composition_relations},
\ref{eq:commutation_relations}) remain valid for $n>1$.

We conclude this section with two formulas that will be important in
what follows and can both be verified immediately. First, the Weyl
operators are trace-less, with the exception of the trivial one:
\begin{eqnarray}
	\mathrm{tr} \left( w(v) \right) &=& d^n \delta_{v,0}.	\label{eq:weyl_trace}
\end{eqnarray}
Second, for any vector $v \in V$ and any
subspace $W \subseteq V$ one has
\begin{eqnarray}
\sum_{w \in W} \omega^{[v,w]} &=& \label{eq:phase_summation}
	\begin{cases}
	|W| & \textrm{if }	[v,w]=0 \; \forall w \in W,		\\
	0 	&	\textrm{else}.
	\end{cases}
\end{eqnarray}

\subsection{Stabilizer States} \label{sub:stabilizer_states}

Here, we will cast the established theory \cite{gottesman_stabilizer_1997, nielsen_quantum_2010} of
stabilizer states into the language of symplectic geometry required
for our proof. For previous similar expositions, see
\cite{gross_hudsons_2006, gross_stabilizer_2013}.

Note that Equation~(\ref{eq:commutation_relations}) implies that two
Weyl operators $w(u)$ and $w(v)$ commute if and only if $\left[ u, v
\right]=0$. 
Now consider the image of an entire subspace $M \subseteq V$
under the Weyl representation. We define 
\begin{equation*}
	w(M) = \left\{ w(m) : \; m \in M \right\}
\end{equation*}
and observe that $w(M)$ consists of mutually commuting operators if
and only if $\left[ m, m' \right] = 0$ holds for all $m , m' \in M$.
Spaces having this property are called \emph{isotropic}. Assume now
that $M$ is isotropic.

If $d$ is odd, then the $w(M)$ not only commute, but actually form a
group $w(u)w(v)=w(u+v)$. That's because in
(\ref{eq:composition_relations}), the phase
factor depends on $[u,v]$ modulo $d$, which is zero by assumption for
$u,v\in M$. For even dimensions, however, $[u,v]$ might equal $d$ and
in this case,
the product $w(u)w(v) = - w(u+v)$ does not lie in
$w(M)$ (in other words, $v \mapsto w(v)$ is only a
\emph{projective} representation of the additive group of $M$).
This would create problems in our analysis below. Fortunately, it
turns out that one can choose phases $c(v)\in\{\pm 1\}$ such that
$v\mapsto c(v) w(v)$ does become a true representation of $M$. We will
now describe this construction.

To this end, choose a basis $\mathcal{B} = \left\{ u_1,\ldots, u_{\dim
M}
\right\}$ of $M$. For a given element $m\in M$, let $m = \sum_i m_i
u_i$ be the expansion of $m$ with respect to this basis.  Define the
(basis-dependent) Weyl operators to be:
\begin{equation}
	w_{\mathcal{B}} (m) := 
	\prod_{i=1} w(u_i)^{m_i}.	
	\label{eq:basis_representation}
\end{equation}
Using the fact that the $w(u_i)$ commute, one then obtains for
$m,m'\in M$
\begin{align*}
	w_{\mathcal{B}}(m) w_{\mathcal{B}}(m') =& \prod_{i=1}^n w(u_i)^{m_i} \prod_{i=1}^n w(u_i)^{m'_i} \\
	=& \prod_{i=1}^n w(u_i)^{m_i + m'_i}
	= w_{\mathcal{B}}(m+m').
\end{align*}
This is the desired representation of $M$.

Stabilizer states turn out to be related to \emph{maximal} isotropic
spaces $M$.
We call a subspace $M \subseteq V$ \emph{Lagrangian} (LAG) -- or
\emph{maximally isotropic} -- if every vector $v \in V$ that commutes
with all elements of $M$ is already contained in $M$. This is
precisely the case if
\begin{equation*}
M = 
 \left\{ v \in V : \; [v,m] = 0 \; \forall m \in M \right\} =: M^\perp,
\end{equation*}
where $M^\perp$ denotes the symplectic complement of $M$.
A basic result of symplectic geometry (e.g.\ Satz 9.11 in
\cite{huppert}) states that this condition is fulfilled if and only if
$\dim M = \frac{1}{2} \dim V = n$, or equivalently $|M| = d^n$.  

We are now ready to state the relation between Lagrangian subspaces
and state vectors in Hilbert space:

\begin{theorem}[Stabilizer States]\label{th:stabilizer_states}
Let $M \subset V$ be a Lagrangian subspace, let $\mathcal{B}$ be a
basis of $M$. Then the following assertions are valid:
\begin{enumerate}
\item Up to a global phase, every $v \in M$ singles out one unit vector $|M,v \rangle \in \mathcal{H}$
-- called a \emph{stabilizer state} that fulfills the eigenvalue equations
\begin{equation}
	\omega^{ [v,m] } w_\mathcal{B}(m) |M,v\rangle = |M,v \rangle
	\quad \forall m \in M.	\label{eq:stabilizer_definition}
\end{equation}
\item Two elements $u,v\in M$ define the same stabilizer state
if and only if they belong to the same affine space
$\left[ v \right]_M := \left\{ v + m, \; m \in M \right\}$ modulo $M$.
If this is not the case, the resulting stabilizer states are orthogonal,
i.e.\ $\langle M, u | M, v \rangle = 0$.
\item 
$V$ can be decomposed into a union of $d^n = \mathrm{dim}(\mathcal{H})$ different affine spaces modulo $M$. Via \eqref{eq:stabilizer_definition}, this union defines an orthonormal basis of stabilizer states associated with $M$.
\end{enumerate}
\end{theorem}

This statement implies that each stabilizer state is uniquely
characterized by a Lagrangian subspace $M \subset V$ and one
particular affine space $[v]_M$ modulo $M$.  In the remainder of this
article it will be convenient to represent each such affine space by a
representative $\zeta \in [v]_M \in V$ contained in it.  We have opted
to denote such representatives of cosets $\zeta, \iota \in V$ by greek
letters to notationally underline their origin.

\begin{proof}[Proof of \autoref{th:stabilizer_states}]

Define
\begin{equation*}
\rho_{M,v} := d^{-n} \sum_{m \in M} \omega ^{ [v,m]}
w_{\mathcal{B}}(m)
\end{equation*}
and compute
\begin{align*}
\rho_{M,v }^2 =& d^{-2n} \sum_{m,m' \in M}	\omega ^{ [v,m]} \omega ^{ [v,m']}w_{\mathcal{B}}(m)w_{\mathcal{B}}(m')	\\
=& d^{-2n} \sum_{m,m' \in M}	\omega ^{ [v,m+m']} w_{\mathcal{B}}(m+m')		\\
=& d^{-n} \sum_{m \in M} \omega^{[v,m]} w_{\mathcal{B}}(m) = \rho_{M, v },
\end{align*}
as well as
\begin{align*}
	\operatorname{tr}\rho_{M,v} 
	=& d^{-n}
	\sum_{m,\in M}	\omega ^{[v,m]} \operatorname{tr}w_{\mathcal{B}}(m) \\
	=& d^{-n} \operatorname{tr} w_{\mathcal{B}}(0) = 1
\end{align*}
where we have employed (\ref{eq:weyl_trace}).
The first relation implies that
$\rho_{M,v}$ is a projection and the second one that is has rank one.
One can check by direct calculation that 
\begin{equation*}
	\omega^{ [v,m] } w_\mathcal{B}(m) 
	\rho_{M,v} 
	=
	\rho_{M,v},
\end{equation*}
holds for every $m \in M$.
Consequently, the 
so that the any vector from the range of 
$\rho_{M,v}$ fulfills all eigenvalue equations. However, since $\rho_{M,v}$ has rank one, its range 
corresponds to a single vector that we can associate with $|M,v \rangle \in \mathcal{H}$
up to a global phase.
This proves the first claim up to uniqueness which we are going to establish later on.

For the second claim, fix $u,v \in V$ and observe
\begin{align*}
\mathrm{tr} \left( \rho_{M, u} \rho_{M, v} \right)
=& d^{-2n} \sum_{m,m' \in M} \omega^{ [u,m] } \omega ^{ [v,m'] }\mathrm{tr} \left( w_{\mathcal{B}}(m+m') \right)	\\
=& d^{-2n} \sum_{m,m' \in M} \omega ^{ [u,m] } \omega ^{ [v,m'] } d^n \delta_{m+m',0}	\\
=& d^{-n} \sum_{m \in M} \omega ^{ [u - v, m] } 	\\
=& d^{-n} \begin{cases}
		|M| & \textrm{if } [u - v,m] = 0 \; \forall m \in M,	\\
		0 & \textrm{else},
		\end{cases}		
\end{align*}
where we have used \eqref{eq:phase_summation}.
But because $M$ is \emph{maximally} isotropic, $[u - v,m] = 0
\; \forall m \in M$ implies $u-v\in M$. Thus, there is one
$\rho_{M,u}$ for each affine space $u + M\subset V$, and two
distinct affine spaces give rise to othogonal states which is just the second claim. 

Finally, note that there are
$|V/M|=d^n=\dim\mathcal{H}$ such affine spaces, which proves that one
obtains an ortho-normal basis in this way. Moreover, this establishes the uniqueness
part of the first statement and implies, 
justifying that $|M,v \rangle$ is well-defined up to a global phase. 
\end{proof}

In the remainder of this section, we will show how to choose
consistent bases for two, possibly intersecting, Lagrangian spaces $M,
N$ and use these results to come up with formulas for the inner
product between two arbitrary stabilizer states.

\begin{lemma}[Compatible bases]	\label{lem:compatible_bases}
Let $M,N \subset V$ be two Lagrangian subspaces.
Then there exists 
bases $\mathcal{B}_M$ of $M$ and $\mathcal{B}_N$ of $N$
such that $w_{\mathcal{B}_K}(m) = w_{\mathcal{B}_M}(m) =
w_{\mathcal{B}_N}(m)$ for any $m \in M \cap N$.  
What is more, for
$m \in M$ and $n \in N$, it holds that
\begin{equation}
\mathrm{tr} \left( w_{\mathcal{B}_M} (m) w_{\mathcal{B}_N} (-n)\right) = d^n \delta_{m,n}.	\label{eq:compatible_bases}
\end{equation}
\end{lemma} 

\begin{proof}
	Choose a basis $\{u_1, \dots, u_{\dim{M\cap N}}\}$ of $M\cap N$.  By
	elementary linear algebra, it can be extended both to a basis
	$\mathcal{B}_M$ of $M$ and to a basis $\mathcal{B}_N$ of $N$. The
	first claim follows immediately from
	(\ref{eq:basis_representation}). For the second claim, note that for
	from (\ref{eq:composition_relations}), we have that
	$w_{\mathcal{B}_M} (m) w_{\mathcal{B}_N} (-n) = \pm w(m-n)$. Thus,
	by (\ref{eq:weyl_trace}), the trace in (\ref{eq:compatible_bases})
	vanishes unless $m=-n$. In that case, however, $m,n\in K$ 
	and thus, by construction of the bases,
	$w_{\mathcal{B}_M}(m)=w_{\mathcal{B}_K}(m)$ and 
	$w_{\mathcal{B}_N}(-n)=w_{\mathcal{B}_K}(-n)$. Thus 
	\begin{equation*}
		w_{\mathcal{B}_M} (m) w_{\mathcal{B}_N} (-n) =
		 w_{\mathcal{B}_K}(m-n) = w_{\mathcal{B}_K}(0) = w(0).
	\end{equation*}
	The claim then follows from (\ref{eq:weyl_trace}).
\end{proof}

We conclude this subsection with an important observation: The overlap
of different stabilizer states is fully characterized by the geometric
intersection of their underlying Lagrangian subspaces.

\begin{lemma}[Overlap of stabilizer states]\label{prop:stabilizer_overlap}
Let $|M, \zeta \rangle,|N,\iota \rangle \in \mathcal{H}$ be two stabilizer states characterized by Lagrangian subspaces $M,N\subset V$ 
(as well as corresponding bases $\mathcal{B}_M$ and $\mathcal{B}_N$ if $d$ is even)
and representatives $\zeta , \iota \in V$ of cosets $[\zeta]_M \in
V/M$ and $[\iota ]_N \in V/N$, respectively. Then, setting $K = M \cap
N$, their inner product is given by 
\begin{equation}
\left| \langle M,\zeta | N,\iota \rangle \right|^2   = \begin{cases}
d^{-n} |K | & \textrm{if } [\zeta, m] = [\iota , m] \; \forall m \in K ,					\\
0 & \textrm{else}.
\end{cases}
\end{equation}
\end{lemma}

\begin{proof}
The claim follows from direct computation. According to \autoref{lem:compatible_bases} we can pick bases $\mathcal{B}_K$ of $K:=M \cap N$, $\mathcal{B}_M$ of $M$ and $\mathcal{B}_N$ of $N$  that are compatible with each other. With respect to these bases we can write
\begin{eqnarray*}
|M,\zeta \rangle \langle M, \zeta | &=& d^{-n} \sum_{m \in M} \omega ^{[\zeta, m]} w_{\mathcal{B}_M}(m),\\
|N, \iota \rangle \langle N, \iota | &=& d^{-n} \sum _{m' \in N} \omega ^{-[\iota ,m']}w_{\mathcal{B}_N}(-m').
\end{eqnarray*}
Formula (\ref{eq:compatible_bases} ) now implies
\begin{align*}
& \left| \langle M,\zeta | N,\iota \rangle \right|^2  \\
=& \mathrm{tr} \left(  |M,\zeta \rangle \langle M,\zeta | |N,\iota \rangle \langle N,\iota  | \right)		\\
=& d^{-2n}   \sum_{m\in M} \sum_{m' \in N} \omega ^{[\zeta,m]-[\iota, m'] } \mathrm{tr} \left(  w_{\mathcal{B}_M}(m) w_{\mathcal{B}_N}(-m') \right)		\\
=& d^{-n} \sum_{m \in M \cap N} \omega  ^{[\zeta-\iota, m]}							\\
=& d^{-n}
\begin{cases}
|M \cap N | & \textrm{if } [\zeta-\iota, m] = 0 \; \forall m \in M \cap N 					\\
0 & \textrm{else},
\end{cases}
\end{align*}
where the last equation follows from formula (\ref{eq:phase_summation}).
\end{proof}

\subsection{Grassmannian subspaces and discrete symplectic geometry} 	

Let $Q$ be a $n$-dimensional vector space over the finite field
$\mathbbm{Z}_d$. 
The \emph{Grassmannian} $\mathcal{G}(d,n,k)$ is the
set of $k$-dimensional subspaces of $V$. A standard result
-- e.g\
formula (9.2.2) in \cite{cameron_combinatorics_1994} -- says that the size of $\mathcal{G}$ is given by the \emph{Gaussian binomial coefficient}: 
\begin{equation}
	|\mathcal{G}(d,n,k)| = \binom nk_d
	 :=		\label{eq:gaussian_binomial}
	\begin{cases}
	\prod_{i=0}^{k-1} \frac{d^{n-i}-1}{d^{k-i}-1} & \textrm{if } k \leq n,	\\
	0 & \textrm{else}.
	\end{cases}
\end{equation}
This is the analogue of the familiar binomial coefficient for the finite field $\mathbb{Z}_d$. As such it exhibits similar properties, such as $\binom{n}{k}_d = \binom{n}{n-k}_d$ (symmetry), $\binom{n}{n}_d = \binom{n}{0}_d = 1$ (trivial coefficients) and Pascal's identity
\begin{equation}
  \binom{n}{k}_d = d^k \binom{n-1}{k}_d + \binom{n-1}{k-1}_d.
	\label{eq:gaussian_pascal1}	\\
\end{equation}
For further reading and proofs of these identities we refer to Chapter 9 in \cite{cameron_combinatorics_1994}
and move on to introducing some core concepts of symplectic geometry:

Let $V$ be a $2n$-dimensional symplectic vector space over the finite field $\mathbb{Z}_d$. 
A \emph{polarization} $(M,N)$ of $V$ is the choice of two Lagrangian 
subspaces $M,N$ which are transverse in the sense that their direct sum spans the entire space, i.e\ $ M \oplus N = V$.
For a fixed Lagrangian $M$ we define the set
\begin{equation*}
	\mathcal{T}(M) = \{ N \,|\, N \text{ Lagrangian; } (M, N) \text{ is
	a polarization of } V\}
\end{equation*}
of all Lagrangian subspaces transverse to $M$. 
The set $\mathcal{T}(M)$
appears in various contexts. For instance it
labels all \emph{graph states} (in a sense explaind below) in quantum information theory \cite{hein_entanglement_2006}

For the purpose of our counting argument, we need to compute the size of
$\mathcal{T}(M) \in V$.

\begin{proposition}\label{prop:polarizations}
	Let $V$ be a $2n$-dimensional symplectic space over $\mathbbm{Z}_d$
	and let $M$ be an arbitrary Lagrangian subspace. 
	Then, the cardinality of $\mathcal{T}(M)$ amounts to
	\begin{equation*}
		\mathcal{T}(d,n)
		:= 
		\big|\mathcal{T}(M)\big| 
		= d^{\frac12 n(n+1)}.
	\end{equation*}
\end{proposition}

\begin{proof}
	Fix $M$ and note that a  subset $N \subset V$ has to be both Lagrangian and transverse to $M$ in order to lie in $\mathcal{T}(M)$.
	These conditions can be made more explicit if we choose a basis $b_1, \ldots, b_{2n}$ of $V$ which obeys 
	\begin{equation*}
	M = \mathrm{span} \left\{ b_1, \ldots , b_n \right\} \quad \textrm{and} \quad \left[b_i, b_j \right] = \delta_{n \oplus i,j} ,
	\end{equation*}
	where
	$\oplus$ denotes addition modulo $2n$.
	Such a basis allows us to fully characterize any subspace $N$ by a $n \times 2n$-generator matrix $G_N$ with column vectors $a_1,\ldots,a_n$ obeying
$
\mathrm{span} \left\{a_1,\ldots,a_n \right\} = N
$.
Moreover, it will be instructive to partition each generator matrix into two $n \times n$ blocks $A$ and $B$, i.e.\ $G_N = \left( 
\begin{array}{c}
A \\
B
\end{array}
\right)$.
Due to our choice of basis the generator matrix $G_M$ of $M$ is particularly simple,
	namely 
	$G_M  = \left( \begin{array}{ccc} \mathbb{I}_{n \times n} & 0_{n \times n} \end{array} \right)^T$.
	Transversality can be restated in terms of these generator matrizes: 
	$M \oplus N = V$ if and only if the $2n \times 2n$-matrix $\left( \begin{array}{ccc} G_M & G_N \end{array} \right) $ has full rank.
	Due to the particular form of $G_M$ this is however equivalent to demanding $\mathrm{rank}(B) = n$.
	Thus we can convert $G_N$ into the equivalent generator matrix 
	$\tilde{G}_N = \left( \begin{array}{ccc} \tilde{A}^T & \mathbb{I}_{n \times n} \end{array} \right)^T$ 
	(and generators $\tilde{a}_1, \ldots \tilde{a}_n$ as above)
	by applying a Gauss-Jordan elimination in the columns of $G_N$. 
	
	The generator matrix $\tilde{G}_N$ characterizes a Lagrangian subspace if and only if $[\tilde{a}_i, \tilde{a}_j ] = 0$ holds for all $i,j=1,\ldots,n$.
	These requirements can be summarized in a single matrix equality, namely
	that $\tilde{G}_N^T J \tilde{G}_N$ must identically vanish. 
	Inserting the particular form of $\tilde{G}_N$ and carrying out the math
	reveals that this is equivalent to demanding 
	that $\tilde{A}^T - \tilde{A}$ must be the zero matrix.
	Hence, a subspace $N$ is a polarization of $M$ if and only if its generator matrix (with respect to the basis chosen above) is Gauss-Jordan equivalent to
	$G_N = \left( \begin{array}{ccc} A & \mathbb{I}_{n \times n} \end{array} \right)^T$,
	where $A$ is a symmetric $n \times n$-matrix over $\mathbb{Z}_d$.
	Therefore there is a one-to-one correspondence between polarizations $N$ of $M$ and symmetric $n \times n$-matrizes over $\mathbb{Z}_d$.
	The dimensionality of the latter is $\frac{1}{2}n(n+1)$ which completes the proof.
\end{proof}

The one-to-one correspondence between polarizations of $M$  and symmetric matrices in this proof gives additional meaning to the set $\mathcal{T}(M)$.
Recall that a stabilizer state $|N, \zeta \rangle$ is a \emph{graph state} if $N$ possesses a generator matrix of the form
$\left( \begin{array}{ccc} A & \mathbb{I}_{n \times n} \end{array} \right)^T$,
where $A$ is a symmetric $n \times n$-matrix.
Hence, $\mathcal{T}(M)$ is the set of all Lagrangian subspaces $N$ which lead to graph states.

The name graph state pays tribute to the fact that $A$ can be interpreted as the adjacency matrix of a (possibly weighted) graph.
Graph states possess a rich structure and many properties of $|N, \zeta \rangle $ can be deduced from the corresponding graph alone.
However, here we content ourselves with pointing out the analogy between graph states and $\mathcal{T}(M)$. 
For further reading we defer the reader to \cite{hein_entanglement_2006}.

Let us now turn to subspaces of the symplectic vector space $V$. It is clear that a proper subspace $W \subset V$ is itself a vector space, however in general it fails to be symplectic. This is due to the fact that the standard symplectic inner product (\ref{eq:symplectic}) of $V$ becomes degenerate if we restrict it to $W$.
Therefore important tools -- such as \autoref{prop:polarizations} -- cannot be directly applied to the proper subspace $W$.
However, this problem can be (partly) circumvent by applying a \emph{linear symplectic reduction}.
For $W \subseteq V$ we define the quotient
\begin{equation}
	\hat W = W/(W^\perp \cap W).	\label{eq:induced_space}
\end{equation}
This space carries the non-degenerate symplectic form 
\begin{equation}\label{eq:induced_form}
	\big[ [v], [w] \big]_{\hat W}:=[v,w]_V	
\end{equation}	
which is easily seen not to depend on the representatives for $[v]$ and $[w]$. 
Consequently, the space $\hat{W}$ endowed with $[\cdot,\cdot]_{\hat{W}}$ is a symplectic vector space.
We will need such a reduction in the proof of \autoref{tm:isotropic_extension}.

\section{Proof of the main Theorem} \label{sec:main_proof}

In this section we show our main result -- \autoref{thm:technical} -- which provides an explicit recursion fully characterizing the frame potential $\mathcal{F}_t (\operatorname{Stabs}(d,n))$ 
of stabilizer states in prime power dimensions $D = d^n$. 
We denote the set of all stabilizer states by $\operatorname{Stabs}(d,n) = \left\{x_1,\ldots,x_{S(d,n)} \right\} \subset \mathbb{C}^D$, where $S(d,n) := |\operatorname{Stabs}(d,n)|$ is just the cardinality of that set.
Recall that in our framework each stabilizer state $x_i \in \mathbb{C}^D$ is specified by a Lagrangian subspace $M$ in $V=\mathbb{Z}_{d}^{2n}$ and a representative $\zeta \in V $ of the coset $[\zeta ]_M \in V/M$.
The Clifford invariance \cite{gross_hudsons_2006} of stabilizer states allows us to calculate any frame potential $\mathcal{F}_t (\operatorname{Stabs}(d,n))$ by counting intersections of Lagrangian subspaces.
This is the content of the following result that considerably simplifies the expression for frame potentials. 

\begin{lemma}	\label{lem:stabilizer_frame_potential}
Let $D =d^n$ be a prime power. The $t$-th frame potential of the set of all stabilizer states in dimension $D$ is given by
\begin{equation}
\mathcal{F}_t (\operatorname{Stabs}(d,n)) = \frac{1}{S(d,n)}\sum_{k=1}^{n}\kappa_M (d,n,k) d^{(1-t)(n-k)},	\label{eq:frame_potential1}
\end{equation}
where  $\kappa_M (d,n,k)$ is the number of Lagrangian subspaces $N$ whose intersection with an arbitrary fixed Lagrangian subspace $M$ is $k$-dimensional. 
\end{lemma}

\begin{proof}

Stabilizer states constitute an orbit of a particular finite unitary group -- the Clifford group. 
Due to this symmetry, the second summation in $\mathcal{F}_t (\operatorname{Stabs}(d,n))$ is 
superfluous and we can write
\begin{align}
\mathcal{F}_t (\operatorname{Stabs}(d,n)) =& \frac{1}{S(d,n)^2} \sum_{i,j=1}^N \left| \langle x_i, x_j \rangle \right|^{2t} \nonumber \\
=& \frac{1}{S(d,n)} \sum_{i=1}^N \left| \langle x_k, x_i \rangle \right|^{2t}, \label{eq:stabilizer_frame_potential_aux1}
\end{align}
where $x_k \in \operatorname{Stabs}(d,n)$ is an arbitrary fixed stabilizer state. \autoref{th:stabilizer_states} assures that any such $x_k$ is unambiguously specified by a Lagrangian subspace $M$ of $V$ and coset $[\zeta ]_M \in M/V$.
Since the choice of $x_k$ in \eqref{eq:stabilizer_frame_potential_aux1} was arbitrary, we can choose $x_k = |M,0 \rangle$ -- i.e.\  it is specified by $M$ and the particularly simple representative $0 \in V$ of the coset $[0]_M$.
Such a choice of $x_k$ together with \autoref{th:stabilizer_states} allows us to rewrite \eqref{eq:stabilizer_frame_potential_aux1} as
\begin{equation}
\mathcal{F}_t (\operatorname{Stabs}(d,n)) = \frac{1}{S(d,n)} \sum_{N \textrm{ LAG}} \sum_{[\zeta ]_N \in V/N} \left| \langle N, \zeta | M, 0 \rangle \right|^{2t}, \label{eq:stabilizer_frame_potential_aux2}
\end{equation}
because instead of summing over stabilizer states, we may as well sum over their characterizing Lagrangian subspaces and cosets instead.
Such a reformulation allows us to employ \autoref{prop:stabilizer_overlap} which implies 
\begin{equation*}
\left| \langle N ,\zeta | M,0 \rangle \right|^{2t}  
= \begin{cases}
d^{-nt}|K |^t& \textrm{if }  [\zeta, m] = 0 \; \forall m \in K, 					\\
0 & \textrm{else},
\end{cases}
\end{equation*}
where $K = M \cap N$ denotes the intersection.
If this intersection is $k$-dimensional, $|K| = d^k$ and consequently
$\left| \langle N, \zeta | M, 0 \rangle \right|^{2t} = d^{-t (n-k)}$, provided that $[\zeta, m ] = 0 $ for all elements $m \in K$. 
This requirement for a non-vanishing overlap is met if and only if $\zeta \in K^\perp$. 
The number of representatives $\zeta$ which obey this property (and single out different stabilizer states) is given by the order of the quotient space $| K^\perp / N|$.
Since $N \subseteq K^\perp$ (which follows from $K \subseteq N$ and $N^\perp = N$), such a quotient space is well defined and its order amounts to
\begin{equation*}
| K^\perp / N | = d^{\dim \left( K^\perp / N \right)} = d^{2n - k - n} = d^{n-k}.
\end{equation*}
Consequently, for each pair of Lagrangians $M,N$ with $k$-dimensional intersection, $d^{n-k}$ out of a total of $d^n$ stabilizer states specified by $N$ give rise to a non-vanishing 
overlap $\left| \langle N, \zeta| M , 0 \rangle \right|^{2t} = d^{-t(n-k)}$ with the fixed stabilizer state $x_k = |M,0 \rangle$. 
Inserting this insight into \eqref{eq:stabilizer_frame_potential_aux2} reveals
\begin{align*}
\mathcal{F}_t (\operatorname{Stabs}(d,n))
=& \frac{1}{S(d,n)} \sum_{N \textrm{ LAG}} d^{(1-t) (n - \dim (N \cap M))} \\
=& \frac{1}{S(d,n)} \sum_{k=1}^N \kappa_M (d,n,k) d^{(1-t)(n-k)},
\end{align*}
where we have replaced the summation over the different Lagrangian subspaces with an equivalent summation over the dimension $k$ of the intersections $M \cap N$. 
\end{proof}

\autoref{lem:stabilizer_frame_potential} shows that we can compute the stabilizer frame potential $\mathcal{F}_{t} (\operatorname{Stabs}(d,n))$ provided that the number $\kappa_M (d,n,k)$ is known for any Lagrangian subspace $M$ 
and any intersection space dimesion $k\in \left\{0,\ldots,n\right\}$. The following two statements characterize that number.

\begin{theorem}	\label{tm:isotropic_extension}
	Let $V$ be a $2n$-dimensional symplectic space over $\mathbbm{Z}_d$.
	Fix an arbitrary Lagrangian subspace $M$ and a $k$-dimensional subspace $K$ of $M$. The number of Lagrangian
	subspaces $N$ that obey $M\cap N = K$
	equals
	\begin{equation*}
		\mathcal{T}(d,n-k)
		=
		d^{\frac12(n-k)(n-k+1)}.
	\end{equation*}
\end{theorem}

The fact that each Lagrangian $M$ admits $|\mathcal{G}(d,n,k)| = \binom{n}{k}_d$
different $k$-dimensional subspaces $K$ (formula (\ref{eq:gaussian_binomial}))  immediately yields the following corollary.

\begin{corollary}[Expression for $\kappa_M(d,n,k)$ ]	\label{cor:kappa}
Let $V$ be a $2n$-dimensional symplectic space over $\mathbbm{Z}_d$.
For an arbitrary Lagrangian subspace $M \subset V$ and $k \in \left\{0,\ldots,n \right\}$, the number of Lagrangian subspaces $N$ whose intersection with $M$ is $k$-dimensional amounts to
\begin{equation}
\kappa_M (d,n,k) = \binom{n}{k}_d  d^{\frac{1}{2}(n-k)(n-k+1)}.
\end{equation}
\end{corollary}

\begin{proof}[Proof of \autoref{tm:isotropic_extension}]

	We need to count in
	how many ways one can choose a Lagrangian space $N\subset V$ that
	intersects $M$ exactly in $K$. Our strategy will be to relate the
	set of such extensions $N$ of $K$ to a set $\mathcal{T}$  as in
	\autoref{prop:polarizations}. 	
	To that end, set $\hat{W} := K^\perp/K$. Note that $K \subseteq K^\perp$ (because $K \subseteq M$ and $M$ is Lagrangian) implies
	\begin{equation*}
		\hat{W}
		= K^\perp / K
		= K^\perp / \big( K \cap K^\perp \big)
		= K^\perp / \big( (K^\perp)^\perp \cap K^\perp \big).
	\end{equation*}
	Therefore $\hat{W}$ is the linear symplectic reduction of $K^\perp$ as defined in (\ref{eq:induced_space}).  
	The space $\hat{W}$ endowed with the induced symplectic product  $[\cdot, \cdot]_{\hat{W}}$  defined in (\ref{eq:induced_form}) 
	forms a symplectic vector space with dimension 
	\begin{equation*}
		\dim \hat{W} = \dim K^\perp / K = 2n-k-k = 2(n-k).
	\end{equation*}
	Note that any isotropic space $N$ containing $K$ is in particular
	contained in $K^\perp$.  The canonical projection $N\mapsto N/K$
	sets up a one-to-one correspondence between $n$-dimensional subspaces
	of $K^\perp$ containing $K$ and $(n-k)$-dimensional subspaces of
	$\hat{W}$. We need two properties of this correspondence:

	(\emph{i}) $N/K\subset \hat{W}$ is isotropic if and only if $N\subset V$
	is. Proof: This follows immediately from (\ref{eq:induced_form}).

	(\emph{ii}) $N/K\subset \hat{W}$ is transverse to $M/K$ if and only if
	$M\cap N=K$. Proof: Basic linear algebra shows
	\begin{equation*}
		(M + N)/K \simeq M/K + N/K.
	\end{equation*}
	For the left hand side:
	\begin{align*}
		\dim(M+N)&=\dim(M)+\dim(N)-\dim(M\cap N) \\
		& \leq 2n-k
	\end{align*}
	with equality if and only if $M\cap N=K$. Hence
	$\dim(M+N)/K\leq2(n-k)$ with the same condition for equality.
	For the right hand side:
	\begin{align*}
		\dim(M/K)+\dim(N/K) 
		&\leq \dim M + \dim N - 2 \dim K \\
		&= 2 (n-k)
	\end{align*}
	with equality if and only if the two spaces are transverse.

	It follows that $M/K$ is a Lagrangian subspace of $\hat{W}$ and there is a
	one-to-one correspondence between Lagrangian spaces $N$ intersecting
	$M$ in $K$ and Lagrangian subspaces of $\hat{W}$ 
	transverse to $M/K$.
	Employing \autoref{prop:polarizations} then yields the desired result.
\end{proof}

Finally, we are going to require an explicit characterization of the
number $S(d,n)$ of stabilizer states. We borrow it from
\cite[Corollary 21]{gross_hudsons_2006}:

\begin{proposition}[Number of stabilizer states]	\label{prop:number_of_stabilizer_states}
For $\mathcal{H}=\left(\mathbb{C}^d\right)^{\otimes n}$, the cardinality $S(d,n)$ of $\operatorname{Stabs}(d,n)\subset \mathcal{H}$ amounts to
\begin{equation}
S(d,n) = |\operatorname{Stabs}(d,n)| = d^n \prod_{j=1}^{n} \left(d^j+1\right)	\label{eq:number_of_stabilizer_states}
\end{equation}
and thus obeys the recursion
\begin{equation}
\frac{S(d,n)}{S(d,n+1)} = \frac{1}{(d^{n+1}+1)d}. \label{eq:cardinality_recursion}
\end{equation}
\end{proposition}

Formula (\ref{eq:number_of_stabilizer_states}) combined with \autoref{cor:kappa} 
allows us to write down the frame potential (\autoref{lem:stabilizer_frame_potential}) explicitly:
\begin{equation}
	\mathcal{F}_t (\operatorname{Stabs}(d,n)) = \frac{1}{S(d,n)} \sum_{k=0}^{n} \binom{n}{k}_d d^{\frac{1}{2}(n-k)(n-k+3-2t)}	\label{eq:final_frame_potential}
\end{equation}
with $S(d,n)$ defined in \eqref{eq:number_of_stabilizer_states}.
Note that this is a purely combinatorical expression that depends solely on $d$ and $n$. 
Analyzing its recursive dependence on $n$ allows us to establish the main result of this work -- \autoref{thm:technical}.

\begin{proof}[Proof of \autoref{thm:technical}]
Let us start with the base case \eqref{eq:base_case} which is readily established. 
Indeed, setting $n=1$ and evaluating formula \eqref{eq:final_frame_potential} reveals 
that for any $d$ and $t$ $\mathcal{F}_t (\operatorname{Stabs}(d,n))$ amounts to
\begin{equation*}
 \frac{1}{(d+1)d} \left( \binom{1}{0}_d d^{\frac{1}{2} (4-2t)} + \binom{1}{1}_d \right) 
= \frac{d^{2-t} +1}{(d+1)d},
\end{equation*}
where we have used $\binom{n}{0}_d = \binom{n}{n}_d = 1$. Let us now move on to establishing the recursive behavior. Replacing $n$ by $(n+1)$ in formula \eqref{eq:final_frame_potential}
and employing Pascal's identity \eqref{eq:gaussian_pascal1} as well as trivial coefficients for Gaussian binomials 
yields
\begin{widetext}
\begin{align}
 \mathcal{F}_t \left(\operatorname{Stabs}(d,n+1)\right) 
=&  \frac{1}{S(d,n+1)}\sum_{k=0}^{n+1} \binom{n+1}{k}_d 
d^{\frac{1}{2} (n+1-k)(n+1-k+3-2t)} \nonumber \\
=&  \frac{1}{S(d,n+1)}
\left( \binom{n+1}{0}_d d^{\frac{1}{2}(n+1)(n+4-2t)}
+ \sum_{k=1}^n  \binom{n+1}{k}_d d^{\frac{1}{2} ( n+1-k)(n-k+4 - 2t)}
+ \binom{n+1}{n+1}_d \right)\nonumber \\
=&   \frac{1}{S(d,n+1)} \left( d^0 \binom{n}{0}_d d^{\frac{1}{2}(n+1)(n+4-2t)}
+ \sum_{k=1}^n \left( d^k \binom{n}{k}_d + \binom{n}{k-1}_d \right) d^{\frac{1}{2} (n+1-k)(n-k+4-2t)}
+ \binom{n}{n}_d \right) \nonumber \\
=&  \frac{1}{S(d,n+1)}
\left(  \sum_{k=0}^n d^k \binom{n}{k}_d d^{\frac{1}{2} (n+1-k)(n-k+4-2t)}
+ \sum_{k=1}^{n+1} \binom{n}{k-1} d^{\frac{1}{2} (n - (k-1))(n-(k-1)+3-2t)} \right),
\label{eq:thm_aux1}
\end{align}
\end{widetext}
where we have encorporated the first and last terms in the first and second summation, respectively. 
Note that the second summation just corresponds to
$\sum_{k=0}^n \binom{n}{k}_d d^{\frac{1}{2} (n-k)(n-k+3-2t)}$ -- which in that very form also appears in \eqref{eq:final_frame_potential}. Importantly, a similar equivalence is true for the first sum appearing in 
\eqref{eq:thm_aux1}. Taking a closer look at the overall exponent of $d$ in that summation
reveals
\begin{align*}
& k + \frac{1}{2} (n+1-k)(n-k+4-2t) \\
=& n - (t-2)  + \frac{1}{2} (n-k)(n-k+3-2t)
\end{align*}
and the first term is independent of the summation index. 
Consequently the first sum in \eqref{eq:thm_aux1} actually corresponds to
$d^{n-(t-2)} \sum_{k=0}^n \binom{n}{k}_d d^{\frac{1}{2}(n-k)(n-k+3-2t)}$
and we can conclude
\begin{align*}
& \mathcal{F}_t ( \operatorname{Stabs}(d,n+1) ) \\
=& \frac{S(d,n)}{S(d,n+1)}
\left( d^{n - (t-2)}+ 1 \right) \\
 \times & \frac{1}{S(d,n)}\sum_{k=0}^n \binom{n}{k}_d d^{\frac{1}{2} (n-k)(n-k+3-2t)} \\
=& \frac{d^{n-(t-2)} + 1}{d (d^{n+1}+1)} 
\frac{1}{S(d,n)} \sum_{k=0}^n \binom{n}{k}_d d^{\frac{1}{2} (n-k)(n-k+3-2t)} \\
=&  \frac{d^{n-(t-2)} + 1}{d (d^{n+1}+1)}  \mathcal{F}_t \left( \operatorname{Stab}(d,n) \right)
\end{align*}
where we have employed \eqref{eq:cardinality_recursion}.
\end{proof}

We conclude this article with presenting a proof of \autoref{cor:stabs3d} which establishes some substantial insights into the structure of stabilizer states.

\begin{proof}[Proof of \autoref{cor:stabs3d}]
	Start with the case $t=2$. Then the result of
	\autoref{thm:technical} reads
	\begin{align*}
		\mathcal{F}_2 (\operatorname{Stabs}(d,1)) 
		=& \frac2{d(d+1)} \\
		\frac{
			\mathcal{F}_2 (\operatorname{Stabs}(d,n+1))
		}{
			\mathcal{F}_2 (\operatorname{Stabs}(d,n))
		}
		=& \frac{ d^n+1 }{ d(d^{n+1}+1) }.
	\end{align*}
	But the Welch Bound (\ref{eq:sidelnikov_inequality})
	satisfies identical relations:
	\begin{align}
		\mathcal{W}_2(d)
		=& \frac2{d(d+1)} \\
		\frac{
			\mathcal{W}_2 (d,n+1)
		}{
			\mathcal{W}_2 (d,n)
		}
		=& 
		\frac{
			\binom{d^n+1}{2}
		}{
			\binom{d^{n+1}+1}{2}
		}
		=
		\frac{
			(d^n+1)d^n
		}{
			(d^{n+1}+1)d^{n+1}
		}
	\end{align}

The 3-design case can be proved along similar lines. We have
\begin{align}
\mathcal{F}_3 (\operatorname{Stabs}(d,1)) =& \frac{1+d^{-1}}{(d+1)d} \label{eq:3design_frame}\\
\frac{\mathcal{F}_3 (\operatorname{Stabs}(d,n+1))}{
\mathcal{F}_3 (\operatorname{Stabs}(d,n))}
= & \frac{d^{n-1}+1}{d(d^{n+1}+1)} \label{eq:frame_3_design}
\end{align}
and the Welch bound satisfies
\begin{align}
\mathcal{W}_3 (d)
=& \binom{d+2}{3}^{-1} = \frac{6}{(d+2)(d+1)d} \label{eq:3design_welch}  \\
\frac{\mathcal{W}_3 (d^{n+1})}{\mathcal{W}_3 (d^n)}
=& \frac{(d^n+2)(d^n+1)}{(d^{n+1}+2)(d^{n+1}+1)d}.  \label{eq:welch_3_design}
\end{align}
The two base values \eqref{eq:3design_frame} and \eqref{eq:3design_welch} coincide for $d \leq 2$.
Otherwise, the former is strictly larger than the latter. 
Comparing the recursion factors yields
\begin{align}
\frac{\textrm{Eq.\ } \eqref{eq:welch_3_design}}
{ \textrm{Eq.\ } \eqref{eq:frame_3_design}}
=& \frac{(d^n+2)(d^n+1)}{(d^{n+1}+2)((d^{n-1}+1)} \\
=& \frac{d^{2n} + 3d^n + 2}{d^{2n} + (2/d + d)d^n +2}
\leq 1
\end{align}
with equality if and aonly if $d=1,2$. 
Consequently we have $\mathcal{F}_3 (\operatorname{Stabs}(d,n)) = \mathcal{W}_3 (d^n)$
for any $n \in \mathbb{N}_+$ if and only if $d \leq 2$.

Finally, let us move on the the 4-design case, where we have
\begin{align}
\mathcal{F}_4 \left( \operatorname{Stabs}(d,1) \right)
=& \frac{1+d^{-2}}{(d+1)d}, \label{eq:4design_base}\\
\frac{\mathcal{F}_4 \left( \operatorname{Stabs}(d,n+1) \right)}{\mathcal{F}_4 \left( \operatorname{Stabs}(d,n) \right)}
=& \frac{d^{n-2}+1}{(d^{n+1}+1)d}, \label{eq:4design_rec} \\
\end{align}
and
\begin{align}
\mathcal{W}_4 (d) =& \frac{24}{(d+3)(d+2)(d+1)d} \label{eq:4welch_base}\\
\frac{\mathcal{W}_4 (d^{n+1})}{\mathcal{W}_4 (d^n)}
=& \frac{(d^n+3)(d^n+2)(d^n+1)}{(d^{n+1}+3)(d^{n+1}+2)(d^{n+1}+1)d}. \label{eq:4welch_rec}
\end{align}
Comparing \eqref{eq:4design_base} to \eqref{eq:4welch_base} reveals
$\mathcal{F}_4 \left( \operatorname{Stabs}(d,1) \right) \geq \mathcal{W}_4 \left( d \right)$ with equality
if and only if $d=1$. An analogous relation holds for \eqref{eq:4design_rec} and \eqref{eq:4welch_rec}
which assures that $\mathcal{F}_4 (\operatorname{Stabs}(d,n) )$ and $\mathcal{W}_4 (d^n)$ 
only ever coincide in the trivial case $d=1$. 

For the final claim of \autoref{cor:stabs3d}, note that the set of
stabilizer states in prime-power dimensions form one orbit under the
action of the Clifford group \cite{gross_hudsons_2006}. Also, any orbit
of a unitary $t$-design is a complex projective $t$-design
\cite{dankert_exact_2009, gross_evenly_2007}. Thus Claim~3 implies that the
Clifford group is not a 4-design. Peter Turner has made us aware of
the fact that the frame potential of group orbits only depends on the
action of that group on the totally symmetric space
$\operatorname{Sym}^t(\mathbb{C}^D)$. Following the reasoning of
\cite{gross_evenly_2007}, a group acting irreducibly on that space
has the property that any orbit constitutes a complex projective
$t$-design. Thus, the stronger statement in Claim~4 is also implied by
Claim~3.
\end{proof}

{\bf Acknowledgements}: 
The authors want to thank P.~Turner for insightful discussions and
H.~Zhu, as well as Z.~Webb for informing us of their impeding work
\cite{zhu2015,webb2015}.

The work of DG and RK is supported by the Excellence Initiative of the
German Federal and State Governments (Grants ZUK 43 \& 81), the ARO
under contract W911NF-14-1-0098 (Quantum Characterization,
Verification, and Validation), the DFG projects GRO 4334/1,2 (SPP1798
CoSIP), and the State Graduate Funding Program of Baden-W\"urttemberg. 

\bibliographystyle{ieeetr}
\bibliography{stabs}

\end{document}